\documentclass[aps,pra,twocolumn,groupedaddress,amsfonts,amssymb,amsmath,showpacs]{revtex4}

\usepackage{amsthm}
\usepackage{graphicx}

\theoremstyle{plain}
\newtheorem{thm}{Theorem}
\newtheorem{proposition}[thm]{Proposition}
\newtheorem{lemma}{Lemma}
\theoremstyle{definition}
\newtheorem{defn}{Definition}

\begin{document}


\title{Equivalence between quantum simultaneous games and quantum
sequential games}


\author{Naoki Kobayashi}
\email{kobayashi@ASone.c.u-tokyo.ac.jp}
\altaffiliation{Contact address: c/o Prof. A. Shimizu, Department of
Basic Science, The University of Tokyo,
3-8-1 Komaba, Tokyo 153-8902, Japan}
\affiliation{Department of Physics, Graduate School of Science, The
University of Tokyo, 7-3-1 Hongo, Bunkyo-ku, Tokyo 113-0033, Japan}


\date{\today}

\begin{abstract}
 A framework for discussing relationships between different types of
 games is proposed. Within the framework, quantum simultaneous games, finite
 quantum simultaneous games, quantum sequential games, and finite quantum
 sequential games are defined. In addition, a notion of equivalence
 between two games is defined. Finally, the following three theorems are
 shown: (1) For any quantum simultaneous game $G$, there exists a
 quantum sequential game equivalent to $G$. (2) For any finite quantum
 simultaneous game $G$, there exists a finite quantum sequential game
 equivalent to $G$. (3) For any finite quantum
 sequential game $G$, there exists a finite quantum simultaneous game
 equivalent to $G$. 
\end{abstract}

\pacs{03.67.-a, 02.50.Le}

\maketitle

\section{Introduction}
Game theory is a well-established branch of applied mathematics 
first developed by von Neumann and Morgenstern \cite{Neumann1944}. It
offers a mathematical 
model of a situation in which decision-makers interact and helps us to
understand what happens in such situations. Although game theory was
originally developed in the context of economics, it has also been
applied to many other disciplines in social sciences like political science
\cite{Ordeshook1986}, and even to biology \cite{Smith1982}.

Meyer \cite{Meyer1999} and Eisert et al. \cite{Eisert1999} brought the game theory into the physics community
and created a new field, quantum game theory. They both quantized a classical
game and found interesting new properties which the original classical game 
does not possess. Nevertheless, their quantized games seem quite different.
$PQ$ penny flipover studied by Meyer is a quantum sequential game, in
which players take turns in performing some operations on a quantum system. 
On the other hand, quantum Prisoners' Dilemma studied by Eisert et al. 
is a quantum simultaneous game, in which there are $n$ players and a
quantum system which consists of $n$ subsystems, and player $i$
performs an operation only on the $i$-th subsystem. 

Since the seminal works of Meyer and Eisert et al., many studies have been
made to quantize classical games and find interesting phenomena
\cite{Eisert2000, Marinatto2000, Benjamin2001, Du2002, Flitney2002, Flitney2004}. Most of the quantum games ever studied are classified
into either quantum simultaneous games or quantum sequential games,
although not much has been done on the latter.

Now that we see that game theory is combined with quantum theory and there
are two types 
of quantum games, several questions naturally arise: (a) Are quantum 
games truly different from classical games? (b) If so, in what sense are they
different? (c) What is the relationship between quantum simultaneous
games and quantum sequential games? 
To answer these questions, it is necessary to examine the whole
structure of game theory including classical games and quantum games,
not a particular phenomenon of a particular game. 

A work by Lee and Johnson \cite{Lee2003} is a study along this line. 
They developed a formalism of games including classical games and quantum
games. With the formalism they addressed the questions (a) and (b),
concluding that 
``playing games quantum mechanically can be more efficient''
and that ``finite classical games consist of a strict subset of finite quantum games''. However, they did not give a precise definition of the phrase
``consist of a strict subset''.

The purpose of this paper is twofold. The first is to present a
foundation and terminology 
for discussing relationships between various types of games.
The second is to answer the question (c). Our conclusions are the following:
(1) For any quantum simultaneous game $G$, there exists a
 quantum sequential game equivalent to $G$. (2) For any finite quantum
 simultaneous game $G$, there exists a finite quantum sequential game
 equivalent to $G$. (3) For any finite quantum
 sequential game $G$, there exists a finite quantum simultaneous game
 equivalent to $G$.

This paper is organized as follows. In Section \ref{framework}, a
framework for describing various types of games and discussing
relationships between them is presented. In Section \ref{varioustypes},
we define within the framework a number of classical and quantum games,
including quantum simultaneous games and quantum sequential games. In
Section \ref{secEquivalence}, we define a notion of `equivalence', which
is one of relationships between two games. Some properties of
equivalence and a sufficient condition for equivalence are also
examined. Section \ref{GameClass} gives the definition of game classes
and some binary relations between game classes. In Section
\ref{MainTheorems}, we prove the three theorems mentioned above.
 Finally, in Section \ref{secDiscussion}, we discuss some consequences
 of the theorems.

\section{\label{framework}Framework for the theory}
In order to discuss relationships between different types of games, we
need a common framework in which various types of games are
described. As the first step in our analysis, we will construct such a
framework for our theory.

For the construction, a good place to start is to consider what is game
theory. Game theory is the mathematical study of \textit{game
situations} which is characterized by the following three features:
\begin{enumerate}
 \item There are two or more decision-makers, or players.
 \item Each player develops his/her strategy for pursuing his/her
       objectives. On the basis of the strategy, he/she chooses his/her  action
       from possible alternatives. 
 \item As a result of all players' actions, some situation is
       realized. Whether the situation is preferable or not for one
       player depends not only on his/her action, but also on the other
       players' actions. 
\end{enumerate}

How much the realized situation is preferable for a player is
quantified by a real number called a payoff. Using this term, we can
rephrase the second feature as ``each player develops his/her strategy
to maximize the expectation value of his/her payoff''. The reason why
the expectation value is used to evaluate strategies is that we can
determine the resulting situation only probabilistically in general,
even when all players' strategies are known. 

As a mathematical representation of the three features of game
situations, we define a normal form of a game.
\begin{defn}
 A normal form of a game is a triplet $(N, \Omega, f)$ whose components
 satisfy the following conditions.
  \begin{itemize}
  \item $N=\{1,2,\dots,n\}$ is a finite set.
  \item $\Omega = \Omega_1\times\dots\times\Omega_n$, where 
	$\Omega_i$ is a nonempty set.
  \item $f$ is a function from $\Omega$ to $\mathbf{R}^n$.
 \end{itemize}
\end{defn}
Here, $N$ denotes a set of players. $\Omega_i$ is a set of player $i$'s
strategies, which prescribes how he/she acts. The $i$-th element
of $f(\omega_1,\dots,\omega_n)$ is the expectation value of the payoff
for player $i$, when player $j$ adopts a strategy $\omega_j$. 

Next, we propose a general definition of games, which works as a
framework for discussing relationship between various kinds of games. We
can regard a game as consisting of some `entities' (like players, cards,
coins, etc.) and a
set of rules under which a game situation occurs. We model the
`entities' in the form of a tuple $T$. Furthermore, we represent the game
situation caused by the `entities' $T$ under a rule $R$ as a normal form
of a game, 
and write it as $R(T)$. Using these formulations, we define a game as follows.
\begin{defn}\label{DefGame}
 We define a game as a pair $(T,R)$, where $T$ is a tuple, and $R$ is a
 rule which determines uniquely a normal form from $T$. When $G=(T,R)$
 is a game, we refer to $R(T)$ as the normal form of the game $G$. We
 denote the set of all games by \textbf{G}.
\end{defn}

The conception of the above definition will be clearer if
we describe various kinds of games in the form of the pair defined above. This
will be done in the next section. 

Thus far, we have implicitly regarded strategies and actions
of \textit{individual} players as elementary components of a game.  In
classical game theory, such modeling of
games is referred to as a noncooperative game, in contrast to a 
cooperative game in which strategies and actions of \textit{groups} of
players are elementary. However, we will call a pair in Definition
\ref{DefGame} simply a game, because
we will deal with only noncooperative games in this paper.

\section{\label{varioustypes}Various types of games}
In this section, various types of classical games and quantum games are
introduced.  First, we confirm
that strategic games, which is a well-established representation of games in classical game theory
(see e.g. \cite{Osborne1994}), can be
described in the 
framework of Definition \ref{DefGame}. Then, we define two quantum
games, namely, quantum simultaneous games and quantum sequential
games. 

\subsection{Strategic Games}
We can redefine strategic games using the framework of Definition
\ref{DefGame} as follows.
\begin{defn}
 A strategic game is a game $(T,R)$ which has the following form.
 \begin{enumerate}
  \item $T=(N,S,f)$, and each component satisfies the following
	condition.
	\begin{itemize}
	 \item $N=\{1,2,\dots,n\}$ is a finite set.
	 \item $S = S_1\times\dots\times S_n$, where 
	       $S_i$ is a nonempty set.
	 \item $f: S\mapsto\mathbf{R}^n$ is a function from $S$ to $\mathbf{R}^n$.
	\end{itemize}
  \item $R(T)=T=(N,S,f)$.
 \end{enumerate}
If the set $S_i$ is finite for all $i$, then we call the game $(T,R)$ a
 finite strategic game. We denote the set of all strategic games by \textbf{SG},
 and the set of all finite strategic games by \textbf{FSG}.
\end{defn}

\begin{defn}
 Let $G=((N,S,f),R)$ be a finite strategic game. Then the mixed extension of
 $G$ is a game $G^*=((N,S,f),R^*)$, where the rule $R^*$ is described
 as follows.
 \begin{itemize}
  \item $R^*(N,S,f)=(N,Q,F)$, where $Q$ and $F$ are of the following forms.
  \item $Q=Q_1\times\dots\times Q_n$, where $Q_i$ is the set of all
	probability distribution over $S_i$. 
  \item $F: Q\mapsto \mathbf{R}^n$ assigns to each $(q_1,\dots,q_n)\in
	Q$ the expected value of $f$. That is, the value of $F$ is given by 
	\begin{align}
	 &F(q_1,\dots,q_n)\nonumber\\
	 &\quad = \sum_{s_1\in S_1}\dots\sum_{s_n\in
	  S_n}\left\{\prod_{i=1}^n q_i(s_i)\right\}f(s_1,\dots,s_n), 
	\end{align}
	where $q_i(s_i)$ is the probability attached to $s_i$. 
 \end{itemize}
 We denote the set of all mixed extensions of finite strategic games by
 \textbf{MEFSG}. 
\end{defn}


\subsection{Quantum Simultaneous Games}
\begin{figure}
 \includegraphics{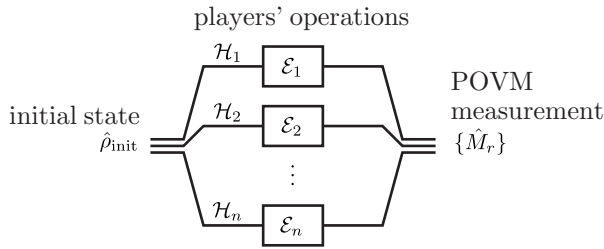}
 \caption{\label{figsim}The setup of a quantum simultaneous game. }
\end{figure}

Quantum simultaneous games are quantum games in which a quantum system
is used according to a protocol depicted in  Fig. \ref{figsim}. In
quantum simultaneous games, there
are $n$ players who can not communicate with each other, and a
referee. The referee prepares a quantum system
in the initial state $\hat{\rho}_{\mathrm{init}}$. The quantum system is
composed of $n$ subsystems, where the Hilbert space for the $i$-th subsystem
is $\mathcal{H}_i$. The referee provides player $i$ with the
$i$-th subsystem. Each player performs some quantum operation on the
provided subsystem. It is determined in advance which operations are
available for each player. After all players finish their operations, they
return the subsystems to the referee. Then the referee performs a POVM
measurement $\{\hat{M}_r\}$ on the total system. If the $r$-th measurement
outcome is obtained, player $i$ receives a payoff $a_r^i$. 

Many studies on quantum simultaneous games have been carried out. Early
significant studies include Refs.
\cite{Eisert1999,Eisert2000,Marinatto2000,Benjamin2001,Du2002}.

The protocol of the quantum simultaneous games is formulated in the form
of Definition \ref{DefGame} as below.
\begin{defn}
 A quantum simultaneous game is a game $(T,R)$ which has the following
 form.
 \begin{enumerate}
  \item
       $T=(N,\mathcal{H},\hat{\rho}_{\mathrm{init}},\Omega,\{\hat{M}_r\},\{\boldsymbol{a}_r\})$,
       and each component satisfies the following condition.
       \begin{itemize}
	\item $N=\{1,2,\dots,n\}$ is a finite set.
	\item
	      $\mathcal{H}=\mathcal{H}_1\otimes\mathcal{H}_2\otimes\dots\otimes\mathcal{H}_n$,
	      where $\mathcal{H}_i$ is a Hilbert space.
	\item $\hat{\rho}_{\mathrm{init}}$ is a density operator on
	      $\mathcal{H}$.
	\item $\Omega=\Omega_1\times\Omega_2\times\dots\times\Omega_n$,
	      where $\Omega_i$ is a subset of the set of all CPTP
	      (completely positive trace preserving) maps
	      on the set of density operators on $\mathcal{H}_i$. In
	      other words, $\Omega_i$ is a set of quantum operations
	      available for player $i$. 
	\item $\{\hat{M}_r\}$ is a POVM on $\mathcal{H}$.
	\item $\boldsymbol{a}_r=(a_r^1,a_r^2,\dots,a_r^n)
	      \in\mathbf{R}^n$. The index $r$ of $\boldsymbol{a}_r$ runs
	      over the same domain as that of $\hat{M}_r$. 
       \end{itemize}
  \item $R(T)=(N,\Omega,f)$. The value of $f$ is given by 
       \begin{equation}
	f(\mathcal{E}_1,\dots,\mathcal{E}_n) = \sum_r\boldsymbol{a}_r\mathrm{Tr}\left[\hat{M}_r(\mathcal{E}_1\otimes\dots\otimes\mathcal{E}_n)(\hat{\rho}_{\mathrm{init}})\right]
       \end{equation}
       for all $(\mathcal{E}_1,\dots,\mathcal{E}_n)\in\Omega$.
 \end{enumerate}
If $\mathcal{H}_i$ is finite dimensional for all $i$,
 then we refer to the game $(T,R)$ as a finite quantum simultaneous
 game. We denote the set of all quantum simultaneous games by \textbf{QSim}, and
 the set of all finite quantum simultaneous games by \textbf{FQSim}.
\end{defn}

\subsection{Quantum Sequential Games}
\begin{figure}
 \includegraphics{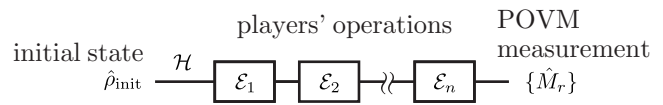}
 \caption{\label{figseq}The setup of a quantum sequential game.}
\end{figure}

Quantum sequential games are another type of quantum games, in which a
quantum system is used according to a protocol depicted in
Fig. \ref{figseq}. In quantum
sequential games, there are $n$ players who can not communicate each
other and a referee. The referee prepares a quantum system in the
initial state $\hat{\rho}_{\mathrm{init}}$. The players performs quantum
operations on the quantum system in turn. The order of the turn may be
regular like $1\to 2\to 3\to 1\to 2\to 3\to\cdots$, or may be
irregular like $1\to 3\to 2\to 3\to 1\to 2\to\cdots$, yet it is
determined in advance. After all the $m$ operations are finished, the
referee performs a POVM measurement $\{\hat{M}_r\}$. If the $r$-th measurement
outcome is obtained, then player $i$ receives a payoff $a_r^i$. 

Games which belong to quantum sequential games include $PQ$ penny
flipover \cite{Meyer1999}, quantum Monty Hall problem \cite{Flitney2002},
and quantum truel \cite{Flitney2004}. 

The protocol of the quantum sequential games is formulated as follows.
\begin{defn}\label{QSeqGame}
 A quantum sequential game is a game $(T,R)$ which has the following
 form.
 \begin{enumerate}
  \item
       $T=(N,\mathcal{H},\hat{\rho}_{\mathrm{init}},Q,\mu,\{\hat{M}_r\},\{\boldsymbol{a}_r\})$,
       and each component satisfies the following condition.
       \begin{itemize}
 	\item $N=\{1,2,\dots,n\}$ is a finite set.
	\item $\mathcal{H}$ is a Hilbert space.
	\item $\hat{\rho}_{\mathrm{init}}$ is a density operator on
	      $\mathcal{H}$.
	\item $Q=Q_1\times Q_2\times\dots\times Q_m$, where $Q_k$ is a
	      subset of the set of all CPTP maps on the set of density
	      operators on $\mathcal{H}$. The total number of operations
	      is denoted by $m$.
	\item $\mu$ is a bijection from $\bigcup_{i=1}^n\{(i,j)|1\le
	      j\le m_i\}$ to $\{1,\dots,m\}$, where $m_i$'s are natural
	      numbers satisfying $m_1+\dots + m_n=m$. The meaning of
	      $\mu$ is that the $j$-th operation for player $i$ is the
	      $\mu(i,j)$-th operation in total.
	\item $\{\hat{M}_r\}$ is a POVM on $\mathcal{H}$. 
	\item $\boldsymbol{a}_r=(a_r^1,a_r^2,\dots,a_r^n)
	      \in\mathbf{R}^n$. The index $r$ of $\boldsymbol{a}_r$ runs
	      over the same domain as that of $\hat{M}_r$. 
       \end{itemize}
  \item $R(T)=(N,\Omega,f)$. The strategy space
       $\Omega=\Omega_1\times\dots\times\Omega_n$ is constructed as 
       \begin{equation}
	\Omega_i=Q_{\mu(i,1)}\times Q_{\mu(i,2)}\times\dots\times Q_{\mu(i,m_i)}.
       \end{equation}
       The value of $f$ is given by 
       \begin{multline}
	f\left((\mathcal{E}_{\mu(1,1)},\dots,\mathcal{E}_{\mu(1,m_1)}),\dots,(\mathcal{E}_{\mu(n,1)},\dots,\mathcal{E}_{\mu(n,m_n)})\right)\\
	= \sum_r \boldsymbol{a}_r\mathrm{Tr}\left[\hat{M}_r \mathcal{E}_m\circ\mathcal{E}_{m-1}\circ\dots\circ\mathcal{E}_{1}(\hat{\rho}_{\mathrm{init}})\right]
       \end{multline}
       for all
       \begin{equation*}
	\left((\mathcal{E}_{\mu(1,1)},\dots,\mathcal{E}_{\mu(1,m_1)}),\dots,(\mathcal{E}_{\mu(n,1)},\dots,\mathcal{E}_{\mu(n,m_n)})\right)
       \end{equation*}
       in $\Omega$.
 \end{enumerate}
If $\mathcal{H}$ is finite dimensional,
 then we refer to the game $(T,R)$ as a finite quantum sequential
 game. We denote the set of all quantum sequential games by \textbf{QSeq}, and
 the set of all finite quantum sequential games by \textbf{FQSeq}.
\end{defn}

\section{\label{secEquivalence}Equivalence of games}
In this section, we define equivalence between two games. The basic idea
is that two games are equivalent if their normal forms have the same
structure, for the essence of a game is a game situation which is
modeled by a normal form. The difficulty of this idea is that a
strategy set $\Omega_i$ may have some redundancy; that is, two or more
elements in $\Omega_i$ may represent essentially the same strategy. If
this is the case, it does not work well to compare the strategy sets
directly to judge whether two games are equivalent or not. Instead, we
should define a new normal form in which the redundancy in the strategy
set is excluded from the original normal form, and then compare the new
normal forms of the two games. 

As the first step to define equivalence between games, we clarify what
it means by ``two elements in $\Omega_i$ represent essentially the same
strategy''.
\begin{defn}
 Let $(N, \Omega, f)$ be a normal form of a game. Two strategies
 $\omega_i,\omega_i'\in\Omega_i$ for player $i$ are said to be
 redundant if
 \begin{align}
  &f(\omega_1\dots\omega_{i-1},\omega_i,\omega_{i+1}\dots\omega_n)\nonumber\\
  &\quad =
    f(\omega_1\dots\omega_{i-1},\omega_i',\omega_{i+1}\dots\omega_n)
 \end{align}
  for all $\omega_1\in\Omega_1, \dots,\omega_{i-1}\in\Omega_{i-1},\
 \omega_{i+1}\in\Omega_{i+1},\dots,\omega_{n}\in\Omega_n$. If two
 strategies $\omega_i,\omega'_i\in\Omega_i$ are redundant, we write
 $\omega_i\sim\omega'_i$. 
\end{defn}

We can show that the binary relation $\sim$ is an \textit{equivalence
relation}. Namely, for all elements $\omega$, $\omega'$, and $\omega''$
of $\Omega_i$, the following holds:
\begin{enumerate}
 \item $\omega\sim\omega$.
 \item If $\omega\sim\omega'$ then $\omega'\sim\omega$.
 \item If $\omega\sim\omega'$ and $\omega'\sim\omega''$ then $\omega\sim\omega''$.
\end{enumerate}

Since $\sim$ is an equivalence relation, we can define the quotient set
$\tilde{\Omega}_i$ of a strategy set $\Omega_i$ by $\sim$. The quotient set
$\tilde{\Omega}_i$ is the set of all equivalence classes in
$\Omega_i$. An equivalence class in $\Omega_i$ is a subset of $\Omega_i$
which has the form of $\{\omega\,|\,\omega\in\Omega_i, a\sim\omega \}$,
where $a$ is an element of $\Omega_i$. We denote by $[\omega]$ an
equivalence class in which $\omega$ is included, and we define
$\tilde{\Omega}$ as
$\tilde{\Omega}\equiv\tilde{\Omega}_1\times\dots\times\tilde{\Omega}_n$.
This $\tilde{\Omega}$ is a new strategy set which has no redundancy.

Next, we define a new expected payoff function $\tilde{f}$ which maps
$\tilde{\Omega}$ to $\mathbf{R}^n$ by 
\begin{equation}
 \tilde{f}([\omega_1],\dots,[\omega_n]) = f(\omega_1,\dots,\omega_n).
\end{equation}
This definition says that for $(C_1,\dots,C_n)\in\tilde{\Omega}$, the
value of $\tilde{f}(C_1,\dots,C_n)$ is determined by taking one element
$\omega_i$ from each $C_i$ and evaluating $f(\omega_1,\dots,\omega_n)$. 

$\tilde{f}$ is well-defined. That is to say, the value of
$\tilde{f}(C_1,\dots,C_n)$ is independent of which element in $C_i$ one
would choose. To show this, suppose $(C_1,\dots,C_n)\in\tilde{\Omega}$
and $\alpha_i,\beta_i\in C_i$. Then $\alpha_i\sim\beta_i$ for every $i$, so that 
\begin{align}
 f(\alpha_1,\alpha_2,\alpha_3,\dots,\alpha_n) &= f(\beta_1,\alpha_2,\alpha_3,\dots,\alpha_n)\\
 &= f(\beta_1,\beta_2,\alpha_3,\dots,\alpha_n)\\
 &\quad\vdots \nonumber\\
 &= f(\beta_1,\beta_2,\beta_3,\dots,\beta_n).
\end{align}
Thus the value of $\tilde{f}(C_1,\dots,C_n)$ is determined uniquely.

Using $\tilde{\Omega}$ and $\tilde{f}$ constructed from the
original normal form $(N,\Omega,f)$, we define the new normal form as follows.
\begin{defn}
 Let $(N,\Omega,f)$ be the normal form of a game $G$. We refer to
 $(N,\tilde{\Omega},\tilde{f})$ as the reduced normal form of $G$. 
\end{defn}

Whether two games are equivalent or not is judged by comparing the
reduced normal forms of these games, as we mentioned earlier.
\begin{defn}\label{equivalent}
 Let $(N^{(1)}, \tilde{\Omega}^{(1)}, \tilde{f}^{(1)})$ be the reduced
 normal form of a game $G_1$, and let $(N^{(2)},\tilde{\Omega}^{(2)},
 \tilde{f}^{(2)})$ be the reduced normal form of a game $G_2$. Then,
 $G_1$ is said to be equivalent to $G_2$ if the following holds.
 \begin{enumerate}
  \item \label{equivcond1}$N^{(1)}=N^{(2)}=\{1,\dots,n\}$.
  \item \label{equivcond2}There exists a sequence $(\phi_1,\dots,\phi_n)$ of bijection $\phi_k:
	 \tilde{\Omega}_k^{(1)}\mapsto\tilde{\Omega}_k^{(2)}$, such that
	for all $(C_1,\dots,C_n) \in
	\tilde{\Omega}^{(1)}$
	\begin{equation}
	 \tilde{f}^{(1)}(C_1, \dots, C_n)
	  = \tilde{f}^{(2)}(\phi_1(C_1),\dots,
	  \phi_n(C_n)). \label{equivcond}
	\end{equation} 
 \end{enumerate}
 If $G_1$ is equivalent to $G_2$, we write $G_1\parallel G_2$. 
\end{defn}

To give an example of equivalent games, let us consider classical $PQ$
penny flipover \cite{Meyer1999}, in which both player $P$ and player $Q$
are classical players. 
In this game, a penny is placed initially heads up in a box. Players
take turns ($Q\to P\to Q$) flipping the penny over or not. Each player
can not know what the opponent did, nor see inside the box. Finally the
box is opened, and $Q$ wins if the penny is heads up. This game can be
formulated as a finite strategic game whose payoff matrix is given in
Table \ref{payoffmatrix}. 

\begin{table}
 \caption{\label{payoffmatrix}Payoff matrix for $PQ$ penny flipover. $F$
 denotes a flipover and $N$ denotes no flipover. The
 first entry in the parenthesis denotes $P$'s payoff and the second one
 denotes $Q$'s payoff.}
 \begin{ruledtabular}
 \begin{tabular}{ccccc}
  & $Q$: $NN$ & $Q$: $NF$ & $Q$: $FN$ & $Q$: $FF$\\
  \hline
  $P$: $N$ & $(-1,1)$ & $(1,-1)$ & $(1,-1)$ & $(-1,1)$\\
  $P$: $F$ & $(1,-1)$ & $(-1,1)$ & $(-1,1)$ & $(1,-1)$
 \end{tabular}
 \end{ruledtabular}
\end{table}

Intuitively, $Q$ does not benefit from the second move, so that it does
not matter whether $Q$ can do the second move or not. The notion of
equivalence captures this intuition; the above penny flipover game is
equivalent to a finite strategic game whose payoff matrix is given in
Table \ref{reducedpenny}. It represents another penny flipover game in which
both players act only once. Proof of the equivalence is easy and we omit
it. 

\begin{table}
 \caption{\label{reducedpenny}Payoff matrix for another $PQ$ penny
 flipover in which both players act only once.}
 \begin{ruledtabular}
  \begin{tabular}{ccc}
   & $Q$: $N$ & $Q$: $F$\\
   \hline
   $P$: $N$ & $(-1,1)$ & $(1,-1)$\\
   $P$: $F$ & $(1,-1)$ & $(-1,1)$
  \end{tabular}
 \end{ruledtabular}
\end{table}

We now return to the general discussion on the notion of equivalence. The following is a basic property of the equivalence between two games.
\begin{lemma}
 The binary relation $\parallel$ is an equivalence relation; namely, for
 any games $G_1$, $G_2$, and $G_3$, the following holds.
 \begin{enumerate}
  \item $G_1\parallel G_1$ (reflexivity).
  \item If $G_1\parallel G_2$, then $G_2\parallel G_1$ (symmetry).
  \item If $G_1\parallel G_2$ and $G_2\parallel G_3$, then $G_1\parallel
	G_3$ (transitivity).
 \end{enumerate}
\end{lemma} 
\begin{proof}
 For the proof, let $(N^{(i)},\tilde{\Omega}^{(i)},\tilde{f}^{(i)})$ be
 the reduced normal form of $G_i$. 

 The reflexivity is evident.

 Let us prove the symmetry. Assume $G_1\parallel G_2$. Then $N^{(1)}=N^{(2)}$, and there exists a
 sequence $(\phi_1,\dots,\phi_n)$ of bijection $\phi_k:
 \tilde{\Omega}_k^{(1)}\mapsto\tilde{\Omega}_k^{(2)}$, such that 
 for all $(C_1,\dots,C_n) \in \tilde{\Omega}^{(1)}$,
 \begin{equation}
  \tilde{f}^{(1)}(C_1, \dots, C_n)
   = \tilde{f}^{(2)}(\phi_1(C_1),\dots,
   \phi_n(C_n)). \label{sym1}
 \end{equation} 
 Then, there exists a sequence $(\phi_1^{-1},\dots,\phi_n^{-1})$ of
 bijection
 $\phi_k^{-1}:\tilde{\Omega}_k^{(2)}\mapsto\tilde{\Omega}_k^{(1)}$, such
 that for any $(D_1,\dots,D_n)\in\tilde{\Omega}^{(2)}$,
 \begin{equation}
  \tilde{f}^{(2)}(D_1,\dots,D_n) = \tilde{f}^{(1)}(\phi_1^{-1}(D_1),\dots,\phi_n^{-1}(D_n)).
 \end{equation}
 Thus, $G_2\parallel G_1$. 

 We proceed to the proof of the transitivity. Assume $G_1\parallel G_2$ and
 $G_2\parallel G_3$. Then $N^{(1)}=N^{(2)}$ and $N^{(2)}=N^{(3)}$, which
 leads to $N^{(1)}=N^{(3)}$. Furthermore, (i) there exists a sequence
 $(\phi_1,\dots,\phi_n)$ of bijection $\phi_k:
 \tilde{\Omega}_k^{(1)}\mapsto\tilde{\Omega}_k^{(2)}$ such that for all
 $(C_1,\dots,C_n) \in \tilde{\Omega}^{(1)}$, 
 \begin{equation}
  \tilde{f}^{(1)}(C_1, \dots, C_n)
   = \tilde{f}^{(2)}(\phi_1(C_1),\dots,
   \phi_n(C_n)), 
 \end{equation} 
 and (ii) there exists a sequence
 $(\psi_1,\dots,\psi_n)$ of bijection $\psi_k:
 \tilde{\Omega}_k^{(2)}\mapsto\tilde{\Omega}_k^{(3)}$ such that for all
 $(D_1,\dots,D_n) \in \tilde{\Omega}^{(2)}$
 \begin{equation}
  \tilde{f}^{(2)}(D_1, \dots, D_n)
   = \tilde{f}^{(3)}(\psi_1(D_1),\dots,\psi_n(D_n)).
 \end{equation}
 Combining the statements (i) and (ii), we obtain the following statement:
 there exists a sequence $(\psi_1\circ\nolinebreak\phi_1,\dots,\psi_n\circ\phi_n)$
 of bijection $\psi_k\circ\phi_k:
 \tilde{\Omega}_k^{(1)}\mapsto\tilde{\Omega}_k^{(3)}$ such that for all
 $(C_1,\dots, C_n)\in\tilde{\Omega}^{(1)}$,
 \begin{align}
  &\tilde{f}^{(1)}(C_1,\dots,C_n) \nonumber\\
  &\quad =
   \tilde{f}^{(3)}(\psi_1\circ\phi_1(C_1),\dots,\psi_n\circ\phi_n(C_n)). 
 \end{align}
 Thus we conclude that $G_1\parallel G_3$. 
\end{proof}

In some cases, we can find that two games are equivalent by
comparing the normal forms of the games, not the reduced normal
forms. In the 
following lemma, sufficient conditions for such cases are presented.
\begin{lemma}\label{lemma}
 Let $(N^{(1)}, \Omega^{(1)}, f^{(1)})$ be the 
 normal form of a game $G_1$, and let $(N^{(2)},\Omega^{(2)},
 f^{(2)})$ be the normal form of a game $G_2$. If the following
 conditions are satisfied, $G_1$ is equivalent to $G_2$:
 \begin{enumerate}
  \item \label{lemmacond1} $N^{(1)}=N^{(2)}=\{1,\dots,n\}$.
  \item \label{lemmacond2} There exists a sequence $(\psi_1,\dots,\psi_n)$ of bijection $\psi_k:
	 \Omega_k^{(1)}\mapsto\Omega_k^{(2)}$, such that
	for all $(\omega_1,\dots,\omega_n) \in\Omega^{(1)}$,
	\begin{equation}
	 f^{(1)}(\omega_1, \dots, \omega_n)
	  = f^{(2)}(\psi_1(\omega_1),\dots,
	  \psi_n(\omega_n)). \label{lemmapsi}
	\end{equation}
 \end{enumerate}
\end{lemma}
\begin{proof}
 We will show that if the above conditions are satisfied, the conditions
 in the definition \ref{equivalent} are also satisfied.

 From the condition \ref{lemmacond1} in the lemma, the condition
 \ref{equivcond1} in the definition \ref{equivalent} is obviously
 satisfied. 
 
 To show that the condition \ref{equivcond2} in the definition
 \ref{equivalent} is also satisfied, we define a map $\phi_i$ from
 $\tilde{\Omega}^{(1)}_i$ to the set of all subsets of $\Omega_i^{(2)}$ as
 \begin{equation}
   \phi_i(C_i) = \{\psi_i(\omega')\,|\,\omega'\in C_i\}.
 \end{equation}
 We will show that $(\phi_1,\dots,\phi_n)$ is a sequence which satisfies
 the condition \ref{equivcond2} in the definition \ref{equivalent}.

 First, we show that the range of $\phi_i$ is a subset of
 $\tilde{\Omega}^{(2)}_i$; that is, for any
 $[\omega_i]\in\tilde{\Omega}^{(1)}_i$ there exists
 $\xi_i\in\Omega^{(2)}_i$ such that
 $\phi_i([\omega_i]) = [\xi_i]$. In fact, $\psi_i(\omega_i)$
 is such a $\xi_i$:
 \begin{equation}
  \phi_i([\omega_i]) = [\psi_i(\omega_i)]. \label{hoge}
 \end{equation}
 Below, we will prove $\phi_i([\omega_i])
 \subset [\psi_i(\omega_i)]$ first, and then prove
 $[\psi_i(\omega_i)]\subset\phi_i([\omega_i])$. 

 To prove $\phi_i([\omega_i]) \subset [\psi_i(\omega_i)]$, we will show
 that an arbitrary element $\sigma_i\in\phi_i([\omega_i])$ satisfies
 $\sigma_i\in [\psi_i(\omega_i)]$. For this purpose, it is sufficient to
 show that 
 $\sigma_i\sim\psi_i(\omega_i)$; that is, for an arbitrary
 $\sigma_k\in\Omega^{(2)}_k$ $(k\neq i)$
 \begin{align}
  &f^{(2)}(\sigma_1,\dots,\sigma_{i-1},\sigma_i,\sigma_{i+1},\dots,\sigma_n)\nonumber\\
   &\quad =
   f^{(2)}(\sigma_1,\dots,\sigma_{i-1},\psi_i(\omega_i),\sigma_{i+1},\dots,\sigma_{n}). \label{fff}
 \end{align}
 Since $\psi_k$ is a bijection,
 there exists $\omega_k\in\Omega_k^{(1)}$ such that
 $\psi_k(\omega_k)=\sigma_k$. In addition, because
 $\sigma_i\in\phi_i([\omega_i])$, there exists
 $\omega'_i\in [\omega_i]$ such that $\sigma_i =
 \psi_i(\omega'_i)$. Thus, 
 \begin{align}
  &
  f^{(2)}(\sigma_1,\dots,\sigma_{i-1},\sigma_i,\sigma_{i+1},\dots,\sigma_n)
  \nonumber\\
  &
  \begin{aligned}
   = 
  f^{(2)}(\psi_1(\omega_1),\dots,&\psi_{i-1}(\omega_{i-1}),\psi_i(\omega'_i),\\
    &\psi_{i+1}(\omega_{i+1}),\dots,\psi_n(\omega_n))
  \end{aligned}
  \nonumber\\
  &=
  f^{(1)}(\omega_1,\dots,\omega_{i-1},\omega'_i,\omega_{i+1},\dots,\omega_n). \label{lemmaeq2}
 \end{align}
 The last equation follows from \eqref{lemmapsi}. Because
 $\omega'_i\in [\omega_i]$ and $\omega_i\in [\omega_i]$, it follows that
 $\omega'_i \sim \omega_i$. Hence, 
 \begin{align}
  \eqref{lemmaeq2} &=
  f^{(1)}(\omega_1,\dots,\omega_{i-1},\omega_i,\omega_{i+1},\dots,\omega_n)\nonumber\\
  &
  \begin{aligned}
  = 
  f^{(2)}(\psi_1(\omega_1),\dots,&\psi_{i-1}(\omega_{i-1}),\psi_i(\omega_i),\\
   &\psi_{i+1}(\omega_{i+1}),\dots,\psi_n(\omega_n))
  \end{aligned}
  \nonumber\\
  &=
  f^{(2)}(\sigma_1,\dots,\sigma_{i-1},\psi_i(\omega_i),\sigma_{i+1},\dots,\sigma_{n}),\label{lemmaeq3}
 \end{align}
 which leads to the conclusion that the equation \eqref{fff} holds for any $\sigma_i\in\phi_i([\omega_i])$.

 Conversely, we can show that
 $[\psi_i(\omega_i)]\subset\phi_i([\omega_i])$. Let $\sigma_i$ be
 an arbitrary element of $[\psi_i(\omega_i)]$. Since $\psi_i$ is a
 bijection, there exists $\omega'_i\in\Omega_i^{(1)}$ such that
 $\psi_i(\omega'_i)=\sigma_i$. For such $\omega'_i$, it holds that
 $\psi_i(\omega'_i)\sim\psi_i(\omega_i)$, because $\psi_i(\omega'_i)\in
 [\psi_i(\omega_i)]$. Hence,
\begin{align}
 &
 f^{(1)}(\omega_1,\dots,\omega_{i-1},\omega'_i,\omega_{i+1},\dots,\omega_n)\nonumber\\
 &
 \begin{aligned}
 =
 f^{(2)}(\psi_1(\omega_1),\dots,&\psi_{i-1}(\omega_{i-1}),\psi_i(\omega'_i),\\
  &\psi_{i+1}(\omega_{i+1}),\dots,\psi_n(\omega_n))
 \end{aligned}
 \nonumber\\
 &
 \begin{aligned}
 =
 f^{(2)}(\psi_1(\omega_1),\dots,&\psi_{i-1}(\omega_{i-1}),\psi_i(\omega_i),\\
  &\psi_{i+1}(\omega_{i+1}),\dots,\psi_n(\omega_n))
 \end{aligned}
 \nonumber\\ 
 &= f^{(1)}(\omega_1,\dots,\omega_{i-1},\omega_i,\omega_{i+1},\dots,\omega_n),
\end{align}
which indicates that $\omega'_i\sim\omega_i$. Thus,
 $\omega'_i\in [\omega_i]$. Therefore, we conclude that if
 $\sigma_i\in [\psi_i(\omega_i)]$, then
 $\sigma_i=\psi_i(\omega'_i)\in\phi_i([\omega_i])$; that is,
 $[\psi_i(\omega_i)]\subset\phi_i([\omega_i])$. 

 We have shown above that $\phi_i$ is a map from
 $\tilde{\Omega}_i^{(1)}$ to $\tilde{\Omega}_i^{(2)}$. The next thing we
 have to show is that $\phi_i$ is a bijection from
 $\tilde{\Omega}_i^{(1)}$ to $\tilde{\Omega}_i^{(2)}$. We will show the
 bijectivity of $\phi_i$ by proving injectivity and surjectivity
 separately.

 First, we show that $\phi_i$ is injective. Suppose
 $[\omega_i],[\omega'_i]\in\tilde{\Omega}_i^{(1)}$ and
 $[\omega_i]\neq [\omega'_i]$.  Because
 $[\omega_i]\neq [\omega'_i]$, it follows that
 $\omega_i\nsim\omega'_i$, so that there exists
 $(\omega_1,\dots,\omega_{i-1},\omega_{i+1},\dots,\omega_n)\in\Omega_1^{(1)}\times\dots\times\Omega_{i-1}^{(1)}\times\Omega_{i+1}^{(1)}\times\dots\times\Omega_n^{(1)}$
 such that
 \begin{align}
  &f^{(1)}(\omega_1,\dots,\omega_{i-1},\omega_i,\omega_{i+1},\dots,\omega_n)\nonumber\\
   &\quad \neq
   f^{(1)}(\omega_1,\dots,\omega_{i-1},\omega'_i,\omega_{i+1},\dots,\omega_n). 
 \end{align}
 For such $(\omega_1,\dots,\omega_{i-1},\omega_{i+1},\dots,\omega_n)$, 
 \begin{align}
  &
  \begin{aligned}
  f^{(2)}(\psi_1(\omega_1),\dots,&\psi_{i-1}(\omega_{i-1}),\psi_i(\omega_i),\\
   &\psi_{i+1}(\omega_{i+1}),\dots,\psi_n(\omega_n))
  \end{aligned}
  \nonumber\\
  &=
  f^{(1)}(\omega_1,\dots,\omega_{i-1},\omega_i,\omega_{i+1},\dots,\omega_n)\nonumber\\
  &\neq
  f^{(1)}(\omega_1,\dots,\omega_{i-1},\omega'_i,\omega_{i+1},\dots,\omega_n)\nonumber\\
  &
  \begin{aligned}
  =
  f^{(2)}(\psi_1(\omega_1),\dots,&\psi_{i-1}(\omega_{i-1}),\psi_i(\omega'_i),\\
   &\psi_{i+1}(\omega_{i+1}),\dots,\psi_n(\omega_n)).
  \end{aligned}
 \end{align}
 This indicates that $\psi_i(\omega_i)\nsim\psi_i(\omega'_i)$. Hence,
 $[\psi_i(\omega_i)]\neq [\psi_i(\omega'_i)]$. Thus,
 using \eqref{hoge}, we conclude that $\phi_i([\omega_i])\neq\phi_i([\omega'_i])$.

 Next, we show that $\phi_i$ is surjective. Let $[\sigma]$ be an
 arbitrary element of $\tilde{\Omega}_i^{(2)}$. Define
 $\omega\in\Omega_i^{(1)}$ as $\omega\equiv\psi_i^{-1}(\sigma)$. Then,
 \begin{equation}
  \phi_i([\omega]) = [\psi_i(\omega)] = [\sigma].
 \end{equation}
 The first equation follows from \eqref{hoge}. Thus, for an arbitrary
 $[\sigma]\in \tilde{\Omega}_i^{(2)}$, there exists
 $[\omega]\in\tilde{\Omega}_i^{(1)}$ such that
 $\phi_i([\omega])=[\sigma]$. 

 Lastly, we show that $(\phi_1,\dots,\phi_n)$ satisfies
 \eqref{equivcond}. For an arbitrary
 $([\omega_1],\dots,[\omega_n])\in\tilde{\Omega}^{(1)}$, 
 \begin{align}
  \tilde{f}^{(1)}([\omega_1],\dots,[\omega_n]) &=
  f^{(1)}(\omega_1,\dots,\omega_n)\label{leq1}\\
  &= f^{(2)}(\psi_1(\omega_1),\dots,\psi_n(\omega_n))\label{leq2}\\
  &= \tilde{f}^{(2)}([\psi_1(\omega_1)],\dots,[\psi_n(\omega_n)])\label{leq3}\\
  &=
  \tilde{f}^{(2)}(\phi_1([\omega_1]),\dots,\phi_n([\omega_n])).\label{leq4}
 \end{align}
 Equations \eqref{leq1} and \eqref{leq3} follow from the definition of
 $\tilde{f}^{(1)}$ and $\tilde{f}^{(2)}$. Equation \eqref{leq2} follows
 from \eqref{lemmapsi}. The last equation follows from \eqref{hoge}. 
\end{proof}

\section{\label{GameClass}Game Classes}
This short section is devoted to explaining game classes and some binary
relations between game classes. These notions simplify the statements of
our main theorems.

First, we define a game class as a subset of \textbf{G}.
We defined previously \textbf{G}, \textbf{SG}, \textbf{FSG},
\textbf{MEFSG}, \textbf{QSim}, \textbf{FQSim}, \textbf{QSeq}, and
\textbf{FQSeq}. All of these are game classes. Note that \textbf{G}
is itself a game class.

Next, we introduce some symbols. Let \textbf{A} and \textbf{B} be game
classes. If for any game $G\in\mathbf{A}$ there 
exists a game $G'\in\mathbf{B}$ such that $G\parallel G'$, then
we write $\mathbf{A} \trianglelefteq \mathbf{B}$. If
$\mathbf{A}\trianglelefteq \mathbf{B}$ and
$\mathbf{B}\trianglelefteq \mathbf{A}$, we write $\mathbf{A}\bowtie\mathbf{B}$. If $\mathbf{A}\trianglelefteq \mathbf{B}$ but
$\mathbf{B}\ntrianglelefteq \mathbf{A}$, we write $\mathbf{A} \vartriangleleft \mathbf{B}$.

Lastly, we prove the following lemma.
\begin{lemma}\label{preorder}
 The binary relation $\trianglelefteq$ is a preorder. Namely, for any
 game classes $\mathbf{A}$, $\mathbf{B}$, and $\mathbf{C}$, the following
 holds.
 \begin{enumerate}
  \item $\mathbf{A}\trianglelefteq\mathbf{A}$ (reflexivity).
  \item If $\mathbf{A}\trianglelefteq\mathbf{B}$ and
	$\mathbf{B}\trianglelefteq\mathbf{C}$, then
	$\mathbf{A}\trianglelefteq\mathbf{C}$ (transitivity). 
 \end{enumerate}
\end{lemma}
\begin{proof}
 The reflexivity is evident. So we concentrate on proving the
 transitivity. 

 Assume $\mathbf{A}\trianglelefteq\mathbf{B}$ and
 $\mathbf{B}\trianglelefteq\mathbf{C}$. Because
 $\mathbf{A}\trianglelefteq\mathbf{B}$, for any $G_a\in\mathbf{A}$ there
 exists $G_b\in\mathbf{B}$ such that $G_a\parallel G_b$. For such $G_b$, 
 there exists $G_c\in\mathbf{C}$
 such that $G_b\parallel G_c$, since $\mathbf{B}\trianglelefteq\mathbf{C}$. Using the transitivity of the relation
 $\parallel$, we conclude that for any $G_a\in\mathbf{A}$ there 
 exists $G_c\in\mathbf{C}$ such that $G_a\parallel G_c$; that is,
 $\mathbf{A}\trianglelefteq\mathbf{C}$. 
\end{proof}

\section{\label{MainTheorems}Main Theorems}
In this section, we examine relationships between game classes
\textbf{QSim}, \textbf{QSeq}, \textbf{FQSim}, and \textbf{FQSeq}.
\begin{thm}
 $\mathbf{QSim}\trianglelefteq\mathbf{QSeq}$.
\end{thm}

\begin{figure}
 \includegraphics{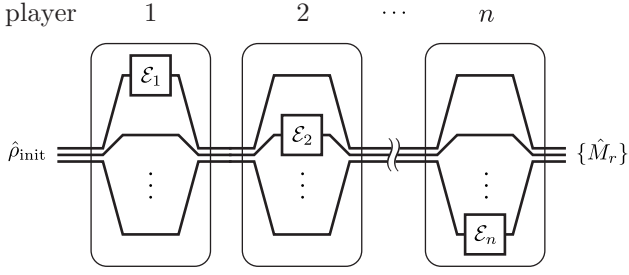}
 \caption{\label{figthm1}A quantum sequential game $G^{\mathrm{seq}}$
 which is equivalent to a quantum simultaneous game $G$ depicted in Fig. \ref{figsim}.}
\end{figure}

\begin{proof}
  We prove the theorem by constructing a quantum sequential game
 $G^{\mathrm{seq}}$ equivalent to a 
 given quantum simultaneous game $G$. We show the construction procedure
 of $G^{\mathrm{seq}}$
 first, and then prove the equivalence using Lemma \ref{lemma}. 
 
 The idea for the construction of $G^{\mathrm{seq}}$ is that a quantum
 simultaneous game depicted in Fig. \ref{figsim} can always be seen as a
 quantum sequential game, as indicated in Fig. \ref{figthm1}. 

 Suppose $G$ is in the following form:
 \begin{gather}
  G=(T,R),\\
  T=(N,\mathcal{H},\hat{\rho}_{\mathrm{init}},\Omega,\{\hat{M}_r\},\{\boldsymbol{a}_r\}),\\
  R(T)=(N,\Omega,f).
 \end{gather}
 Furthermore, suppose $G^{\mathrm{seq}}$ to be constructed is in the
 following form:
 \begin{gather}
  G^{\mathrm{seq}}=(T^{\mathrm{seq}},R^{\mathrm{seq}}),\\
  T^{\mathrm{seq}}=(N^{\mathrm{seq}},\mathcal{H}^{\mathrm{seq}},\hat{\rho}^{\mathrm{seq}}_{\mathrm{init}},Q^{\mathrm{seq}},\mu^{\mathrm{seq}},\{\hat{M}_r^{\mathrm{seq}}\},\{\boldsymbol{a}_r^{\mathrm{seq}}\}),\\
  R^{\mathrm{seq}}(T^{\mathrm{seq}})=(N^{\mathrm{seq}}, \Omega^{\mathrm{seq}}, f^{\mathrm{seq}}).
 \end{gather}
 We construct each component of $T^{\mathrm{seq}}$ from $G$ as follows.
 \begin{itemize}
  \item $N^{\mathrm{seq}}=N=\{1,2,\dots,n\}$.
  \item $\mathcal{H}^{\mathrm{seq}} = \mathcal{H} =
	\mathcal{H}_1\otimes\mathcal{H}_2\otimes\dots\otimes\mathcal{H}_n$.
  \item $\hat{\rho}_{\mathrm{init}}^{\mathrm{seq}}=\hat{\rho}_{\mathrm{init}}$.
  \item $Q^{\mathrm{seq}}=Q_1^{\mathrm{seq}}\times
	Q_2^{\mathrm{seq}}\times\dots\times Q_n^{\mathrm{seq}}$, where 
	\begin{equation*}
	 Q_i^{\mathrm{seq}} =
	\{\mathcal{I}\otimes\dots\otimes\mathcal{I}\otimes(\mathcal{E})_i\otimes\mathcal{I}\otimes\dots\otimes\mathcal{I}\,|\,\mathcal{E}\in\Omega_i\}.
	\end{equation*}
	Here, $\mathcal{I}$ is the identity superoperator.
  \item $\mu^{\mathrm{seq}}$ is a map from $\{(i,1)\,|\,1\le i\le n\}$ to
	$\{1,2,\dots,n\}$. The value of $\mu^{\mathrm{seq}}$ is defined by
	$\mu^{\mathrm{seq}}(i,1)=i$. 
  \item $\{\hat{M}_r^{\mathrm{seq}}\}=\{\hat{M}_r\}$.
  \item $\{\boldsymbol{a}_r^{\mathrm{seq}}\}=\{\boldsymbol{a}_r\}$. 
\end{itemize}
Note that $\Omega_i^{\mathrm{seq}}=Q_i$ because of the construction of
 $Q^{\mathrm{seq}}$ and $\mu^{\mathrm{seq}}$. 

Next, we prove that $G^{\mathrm{seq}}$ constructed above is equivalent
 to $G$, using Lemma \ref{lemma}. From $N^{\mathrm{seq}}=N$, condition
 \ref{lemmacond1} of the Lemma \ref{lemma} is satisfied. To show that condition
 \ref{lemmacond2} of the Lemma \ref{lemma} is also satisfied, we define a map
 $\psi_i: \Omega_i\mapsto\Omega_i^{\mathrm{seq}}$ by
 \begin{equation}
  	\psi_i(\mathcal{E})\equiv
	\mathcal{I}\otimes\dots\otimes\mathcal{I}\otimes
	(\mathcal{E})_i\otimes\mathcal{I}\otimes\dots\otimes\mathcal{I}.
 \end{equation}
 Then, $\psi_i$ is clearly a bijection. Furthermore, for any
 $(\mathcal{E}_1,\dots,\mathcal{E}_n)\in\Omega$, 
 \begin{equation}
  f(\mathcal{E}_1,\dots,\mathcal{E}_n) =
   f^{\mathrm{seq}}(\psi_1(\mathcal{E}_1),\dots,\psi_n(\mathcal{E}_n)).
 \end{equation}
 Thus condition \ref{lemmacond2} of the Lemma \ref{lemma} is satisfied. 
\end{proof}

From the above proof, we can easily see that the following theorem is
also true. 
\begin{thm}\label{Theorem2}
 $\mathbf{FQSim}\trianglelefteq\mathbf{FQSeq}$.
\end{thm}

The converse of Theorem \ref{Theorem2} is the following theorem.

\begin{thm}\label{Theorem3}
 $\mathbf{FQSeq}\trianglelefteq\mathbf{FQSim}$.
\end{thm}

\begin{figure}
 \includegraphics{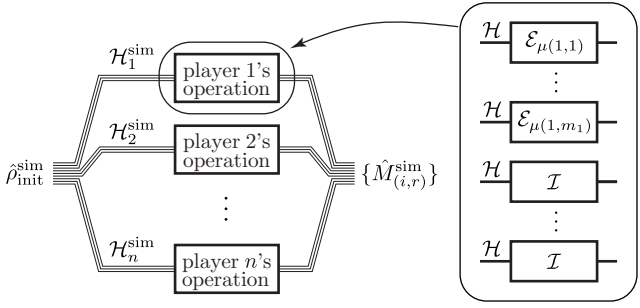}
 \caption{\label{figthm3}A finite quantum simultaneous game $G^{\mathrm{sim}}$
 which is equivalent to a given finite quantum sequential game $G$.}
\end{figure}

\begin{proof}
  We prove the theorem by constructing a finite quantum simultaneous game
 $G^{\mathrm{sim}}$ equivalent to a 
 given finite quantum simultaneous game $G$. We show the construction
 procedure of $G^{\mathrm{sim}}$ first, and then prove the equivalence
 between $G^{\mathrm{sim}}$ and $G$. 
 
 Suppose $G$ is in the following form:
 \begin{gather}
  G=(T,R),\\
  T=(N,\mathcal{H},\hat{\rho}_{\mathrm{init}},Q,\mu,\{\hat{M}_r\}_{r\in \mathcal{R}},\{\boldsymbol{a}_r\}),\\
  R(T)=(N,\Omega,f).
 \end{gather}
 From the above $G$, We construct a finite quantum simultaneous game
 $G^{\mathrm{sim}}$ which is in the following form:
 \begin{gather}
  G^{\mathrm{sim}}=(T^{\mathrm{sim}},R^{\mathrm{sim}}),\\
  T^{\mathrm{sim}}=(N^{\mathrm{sim}},\mathcal{H}^{\mathrm{sim}},\hat{\rho}^{\mathrm{sim}}_{\mathrm{init}},\Omega^{\mathrm{sim}},\{\hat{M}_{(i,r)}^{\mathrm{sim}}\},\{\boldsymbol{a}_{(i,r)}^{\mathrm{sim}}\}),\\
  R^{\mathrm{sim}}(T^{\mathrm{sim}})=(N^{\mathrm{sim}}, \Omega^{\mathrm{sim}}, f^{\mathrm{sim}}).
 \end{gather}
 Figure \ref{figthm3} indicates the setting for $G^{\mathrm{sim}}$. The
 precise instruction on how to construct each component of
 $T^{\mathrm{sim}}$ is given below:
 \begin{itemize}
  \item $N^{\mathrm{sim}}=N=\{1,2,\dots,n\}$.
  \item
       $\mathcal{H}^{\mathrm{sim}}=\mathcal{H}^{\mathrm{sim}}_1\otimes\mathcal{H}^{\mathrm{sim}}_2\otimes\dots\otimes\mathcal{H}^{\mathrm{sim}}_n$, where
       $\mathcal{H}^{\mathrm{sim}}_i = \mathcal{H}^{\otimes
       2m_i}$ and $m_i$ is the number explained in Definition
       \ref{QSeqGame}. This construction means that in game
       $G^{\mathrm{sim}}$, the referee provides player $i$ with a subsystem
       which is itself composed of the $2m_i$ subsystems, each of which
	is the same system as the one used in the original quantum
	sequential game $G$. We write
       the Hilbert space of the $j$-th subsystem of $2m_i$ subsystems
	for player $i$ as
	$\mathcal{H}^{\mathrm{sim}}_{(i,j)}$. Likewise,
       we write a state vector in $\mathcal{H}^{\mathrm{sim}}_{(i,j)}$
       as $|\psi\rangle_{(i,j)}$ and a operator on
       $\mathcal{H}^{\mathrm{sim}}_{(i,j)}$ as $\hat{A}_{(i,j)}$. 
  \item Define a map $\nu:
	\{1,2,\dots,m\}\mapsto\bigcup_{i=1}^n\{(i,m_i+j)\,|\, 1\le j\le
	m_i\}$ by
	\begin{equation}
	 \nu(k)\equiv \mu^{-1}(k)+ (0,m_{i(k)}),
	\end{equation}
	where $i(k)$ is the first element of $\mu^{-1}(k)$. In addition,
	let $\{|1\rangle_{(i,j)},\dots,|d\rangle_{(i,j)}\}$ be an
	orthonormal basis of $\mathcal{H}^{\mathrm{sim}}_{(i,j)}$, where
	$d$ is the dimension of $\mathcal{H}$. Then,
	$\hat{\rho}^{\mathrm{sim}}_{\mathrm{init}}$ is constructed as 
	\begin{widetext}
	\begin{align}
	 	 \hat{\rho}^{\mathrm{sim}}_{\mathrm{init}} =
	 \left(\hat{\rho}_{\mathrm{init}}\right)_{\mu^{-1}(1)}&\otimes
	 |1\rangle_{\nu(m)}\langle 1|_{\nu(m)} \nonumber\\
	 &\bigotimes_{a=1}^{m-1}
	 \left(\frac{1}{\sqrt{d}}\sum_{i_a=1}^d|i_a\rangle_{\nu(a)}|i_a\rangle_{\mu^{-1}(a+1)}\right)
	 \left(\frac{1}{\sqrt{d}}\sum_{j_a=1}^d\langle
	 j_a|_{\nu(a)}\langle j_a|_{\mu^{-1}(a+1)}\right).
	\end{align}
	\end{widetext}
  \item
       $\Omega^{\mathrm{sim}}=\Omega^{\mathrm{sim}}_1\times\Omega^{\mathrm{sim}}_2\times\dots\times\Omega^{\mathrm{sim}}_{n}$,
       where
       \begin{multline}
	\Omega_i^{\mathrm{sim}}=\{\mathcal{E}_{\mu(i,1)}\otimes\mathcal{E}_{\mu(i,2)}\otimes\dots\otimes\mathcal{E}_{\mu(i,m_i)}\otimes\mathcal{I}\otimes\mathcal{I}\otimes\dots\otimes\mathcal{I}\\
	|\,\mathcal{E}_{\mu(i,1)}\in Q_{\mu(i,1)},\dots,
	\mathcal{E}_{\mu(i,m_i)}\in Q_{\mu(i,m_i)}\}.
       \end{multline}
  \item In the game $G^{\mathrm{sim}}$, measurement outcomes are
	described by a pair of variables $(i,r)$, where $i$ takes the value of
	1 or 2, and $r$ is an element of $\mathcal{R}$ (the index set of
	the POVM in the original game $G$). Corresponding POVM elements
	are defined by
	\begin{align}
	 \hat{M}^{\mathrm{sim}}_{(1,r)} &= \hat{K}\otimes (\hat{M}_r)_{\mu^{-1}(m)}\otimes
	 \hat{I}_{\nu(m)},\\
	 \hat{M}_{(2,r)}^{\mathrm{sim}} &= (\hat{I}-\hat{K})\otimes
	(\hat{M}_r)_{\mu^{-1}(m)}\otimes \hat{I}_{\nu(m)}.
	\end{align}
	Here, $\hat{K}$ is defined by
	\begin{widetext}
	\begin{equation}
	 \hat{K}\equiv\bigotimes_{b=1}^{m-1}\left(\frac{1}{\sqrt{d}}\sum_{k_b=1}^d|k_b\rangle_{\mu^{-1}(b)}|k_b\rangle_{\nu(b)}\right)\left(\frac{1}{\sqrt{d}}\sum_{l_b=1}^d\langle l_b|_{\mu^{-1}(b)}\langle l_b|_{\nu(b)}\right).
	\end{equation}
	\end{widetext}
	$\hat{M}_{(1,r)}^{\mathrm{sim}}$ and $\hat{M}_{(2,r)}^{\mathrm{sim}}$ turn out to be positive operators, if we note
	that $\hat{K}$ and $\hat{I}-\hat{K}$ are projection operators
	and $\hat{M}_r$ is a 
	positive operator. Furthermore, the completeness condition is
	satisfied: 
	\begin{equation}
	 \sum_{i=1}^2\sum_{r\in\mathcal{R}} \hat{M}^{\mathrm{sim}}_{(i,r)}=\hat{I}.
	\end{equation}
	Thus, it is confirmed that
	$\{\hat{M}^{\mathrm{sim}}_{(i,r)}\}$ is 
	a POVM. 
  \item We set $\boldsymbol{a}^{\mathrm{sim}}_{(i,r)}$ as
	\begin{equation}
	 \boldsymbol{a}_{(i,r)}^{\mathrm{sim}}=\begin{cases}
			 d^{2m-2}\boldsymbol{a}_r & \text{if } i=1,\\
			 0 & \text{if } i=2.
			\end{cases}
	\end{equation}
 \end{itemize}
 
 Next, we prove that $G^{\mathrm{sim}}$ constructed above is equivalent
 to $G$, using Lemma \ref{lemma}. From $N^{\mathrm{sim}}=N$, condition
 \ref{lemmacond1} of the lemma is satisfied. To show that condition
 \ref{lemmacond2} of the lemma is also satisfied, we define a map
 $\psi_i: \Omega_i^{\mathrm{sim}}\mapsto\Omega_i$ by
 \begin{align}
  &\psi_i(\mathcal{E}_{\mu(i,1)}\otimes\mathcal{E}_{\mu(i,2)}\otimes\dots\otimes\mathcal{E}_{\mu(i,m_i)}\otimes\mathcal{I}\otimes\dots\otimes\mathcal{I})\nonumber\\
  & \quad= (\mathcal{E}_{\mu(i,1)},\mathcal{E}_{\mu(i,2)},\dots,\mathcal{E}_{\mu(i,m_i)}).
 \end{align}
 Then, $\psi_i$ is a bijection. Furthermore, for any element 
 \begin{align}
  &\left((\mathcal{E}_{\mu(1,1)}\otimes\dots\otimes\mathcal{E}_{\mu(1,m_1)}\otimes\mathcal{I}\otimes\dots\otimes\mathcal{I}),\dots,\right.\nonumber\\
  &\qquad\left. (\mathcal{E}_{\mu(n,1)}\otimes\dots\otimes\mathcal{E}_{\mu(n,m_n)}\otimes\mathcal{I}\otimes\dots\otimes\mathcal{I})\right),
 \end{align}
 of $\Omega^{\mathrm{sim}}$,
 one can show after a bit of algebra that 
 \begin{align}
  & f^{\mathrm{sim}}\left( (\mathcal{E}_{\mu(1,1)}\otimes\dots\otimes\mathcal{E}_{\mu(1,m_1)}\otimes\mathcal{
I}\otimes\dots\otimes\mathcal{I}),\dots,\right.\nonumber\\
  &\qquad\left.(\mathcal{E}_{\mu(n,1)}\otimes\dots\otimes\mathcal{E}_{\mu(n,m_n)}\otimes\mathcal{I}\otimes\dots\otimes\mathcal{I})\right) \nonumber\\
  &= f\left(\psi_1(\mathcal{E}_{\mu(1,1)}\otimes\dots\otimes\mathcal{E}_{\mu(1,m_1)}\otimes\mathcal{I}\otimes\dots\otimes\mathcal{I}),\dots,\right.\nonumber\\
   &\qquad\left.\psi_n(\mathcal{E}_{\mu(n,1)}\otimes\dots\otimes\mathcal{E}_{\mu(n,m_n)}\otimes\mathcal{I}\otimes\dots\otimes\mathcal{I})\right).
 \end{align}
 Thus, condition \ref{lemmacond2} of Lemma \ref{lemma} is satisfied.
\end{proof}

\section{\label{secDiscussion}Discussion}
Using Theorem \ref{Theorem2} and Theorem \ref{Theorem3}, we can deduce a
statement about \textbf{FQSeq} (\textbf{FQSim}) from a statement about
\textbf{FQSim} (\textbf{FQSeq}). More precisely, when a statement ``if
$G \in \mathbf{FQSim}$ then $G$ has a property $P$'' is true, another
statement ``if $G\in\mathbf{FQSeq}$ then $G$ has a property $P$'' is
also true, and vice versa. Here, $P$ must be such a property that if a
game $G$ has the property $P$ and $G \parallel G'$, then $G'$ also has
the property $P$. We call such $P$ a property preserved under $\parallel$. For
example, ``a Nash equilibrium exists'' is a property preserved under
$\parallel$. 

Unfortunately, no results are known which have the form ``if
$G\in\mathbf{FQSim}$ (\textbf{FQSeq}) then $G$ has a property $P$, but
otherwise $G$ does not necessarily have the property $P$''. Consequently,
we cannot reap the benefits of the above-mentioned deduction. However,
numerous results exist which have the form ``for a certain subset
$\mathbf{S}$ of \textbf{FQSim} (\textbf{FQSeq}), if $G\in\mathbf{S}$
then $G$ has a property $Q$ preserved under $\parallel$, but otherwise
$G$ does not necessarily have 
the property $Q$''. For such \textbf{S} and $Q$, Theorem \ref{Theorem2}
and Theorems \ref{Theorem3} guarantee that there exists a subset $\mathbf{S'}$
of \textbf{FQSeq} (\textbf{FQSim}) which satisfies the following: ``If
$G\in\mathbf{S'}$ then $G$ has the property $Q$, but otherwise $G$ does
not necessarily have the property $Q$''. In this sense, many of the results
so far on \textbf{FQSim} (\textbf{FQSeq}) can be translated into
statements on \textbf{FQSeq} (\textbf{FQSim}). 

It is worth noting that efficiency of a game \footnote{Amount of
information exchange between players and a referee, required to play a
game. See Ref. \cite{Lee2003} for more detail.} is not a property preserved
under $\parallel$. A good example is in the proof of Theorem \ref{Theorem3}. In
an original quantum sequential game $G$, it is necessary to transmit a
qudit $m+1$ times, while $4m$ times are needed in the
constructed game $G^{\mathrm{sim}}$. Thus, $G^{\mathrm{sim}}$ is far
more inefficient than $G$, despite $G^{\mathrm{sim}}$ and $G$ are
equivalent games. 

Relevant to the present paper is the study by Lee and Johnson
\cite{Lee2003}. To describe their argument, we have to introduce a new
game class.
\begin{defn}
 A finite quantum simultaneous game with all CPTP maps available 
 is a subclass of finite quantum simultaneous games, in which a
 strategy set $\Omega_i$ is the set of all CPTP maps on the set of
 density operators on $\mathcal{H}_i$ for every $i$. We denote the set
 of all finite quantum simultaneous game with all CPTP maps available by
 \textbf{FQSimAll}. 
\end{defn}
 We can easily prove that
$\mathbf{FQSimAll}\vartriangleleft\mathbf{FQSim}$, by showing that the
range of expected payoff functions for a game in 
\textbf{FQSimAll} must be connected, while the one for a game in
\textbf{FQSim} can be disconnected. 

Lee and Johnson claimed that (i) ``any game could be played
classically'' and (ii) ``finite classical games consist of a strict subset
of finite quantum games''.
Using the terms of this paper, we may interpret these claims as follows. 
\begin{thm}\label{prop1}
 $\mathbf{SG}\bowtie\mathbf{G}$.
\end{thm}
\begin{thm}\label{prop2}
 $\mathbf{MEFSG}\trianglelefteq\mathbf{FQSimAll}$.
\end{thm}
\begin{proposition}\label{prop3}
 $\mathbf{FQSimAll}\ntrianglelefteq\mathbf{MEFSG}$.
\end{proposition}

We can prove Theorem \ref{prop1} and Theorem \ref{prop2}, 
regardless of whether or not our interpretations of the claims (i) and
(ii) are
correct. In contrast, we have not yet proven Proposition \ref{prop3}, which is
the reason why we call it a proposition. Nonetheless
Lee and Johnson gave a proof of the statement (ii), so that if the
interpretation that the statement (ii) means 
Theorem \ref{prop2} and Proposition \ref{prop3} is correct, Proposition
\ref{prop3} will be a theorem. 

Using Lemma \ref{preorder} and assuming that Proposition \ref{prop3} is true,
relationships between various game classes can be summarized as follows:
\begin{align}
 \mathbf{MEFSG}\vartriangleleft &\mathbf{FQSimAll}\vartriangleleft \mathbf{FQSim}\bowtie\mathbf{FQSeq}\nonumber\\
 &\trianglelefteq\mathbf{QSim}\trianglelefteq\mathbf{QSeq}\trianglelefteq\mathbf{SG}\bowtie\mathbf{G}. \label{order}
\end{align}
Replacing $\trianglelefteq$ in \eqref{order} with either
$\vartriangleleft$ or $\bowtie$ will be a possible extension of this
paper.

Besides that, it remains to be investigated what the characterizing
features of each game class in \eqref{order} are. Especially, further
research on games which is in \textbf{FQSim} (or equivalently
\textbf{FQSeq}) but not equivalent to any games in \textbf{MEFSG} would
clarify the truly quantum mechanical nature of quantum games.

\begin{acknowledgments}
 I would like to thank Prof. Akira Shimizu for his helpful advice.
\end{acknowledgments}

\bibliography{Qgame}

\end{document}